\documentclass[12pt,a4paper]{article}

\usepackage[dvipsnames]{xcolor}
\usepackage{amsmath}
\usepackage{amsfonts}
\usepackage{amssymb}
\usepackage{amsthm}
\usepackage{array}
\usepackage{hyperref}
\usepackage{url}
\usepackage{graphicx}
\usepackage{listings}
\usepackage{algorithm}
\usepackage{xypic}

\usepackage{biblatex}
\DeclareMathOperator{\Hom}{Hom}

\DeclareMathOperator{\Aut}{Aut}
\DeclareMathOperator{\Gal}{Gal}

\DeclareMathOperator{\Ind}{Ind}
\DeclareMathOperator{\Res}{Res}

\DeclareMathOperator{\Cl}{Cl}

\DeclareMathOperator{\id}{id}

\DeclareMathOperator{\Stab}{Stab}
\DeclareMathOperator{\Ker}{Ker}
\DeclareMathOperator{\im}{Im}
\DeclareMathOperator{\units}{units}
\DeclareMathOperator{\ram}{ram}
\DeclareMathOperator{\ur}{ur}
\DeclareMathOperator{\Sel}{Sel}
\DeclareMathOperator{\val}{val}

\definecolor{darkWhite}{rgb}{0.94,0.94,0.94}

\newcounter{introthm}

\newtheorem{introtheorem}[introthm]{Theorem}

\newtheorem{tm}{Theorem}[section]
\newtheorem{pro}[tm]{Proposition}
\newtheorem{lm}[tm]{Lemma}

\theoremstyle{definition}
\newtheorem{df}[tm]{Definition}
\newtheorem{req}[tm]{Remark}

\newtheorem{algo}[tm]{Algorithm}

\usepackage{tikz}
\usepackage{pgf}
\usetikzlibrary{arrows,automata}
\usetikzlibrary{positioning}
\usetikzlibrary{trees}
\usepackage{comment}
\usepackage{ textcomp }

\usepackage{tkz-tab}

\newcommand{\Q}{\mathbb{Q}}
\newcommand{\Z}{\mathbb{Z}}

\renewcommand{\O}{\mathcal{O}}

\newcommand{\HH}{\mathcal{H}}

\newcommand{\G}{\mathcal{G}}

\newcommand{\LL}{\mathcal{L}}
\newcommand{\K}{\overline{K}^\times}

\date{}
\author{Fabrice Etienne}
\title{An algorithm to compute Selmer groups via resolutions by permutations modules}

\begin{document}

\maketitle

\abstract{Given a number field with absolute Galois group $\G$, a finite Galois module $M$, and a Selmer system $\LL$, this article gives a method to compute $\Sel_\LL$, the Selmer group of $M$ attached to $\LL$. First we describe an algorithm to obtain a resolution of $M$ where the morphisms are given by Hecke operators. Then we construct another group $H_S^1(\G,M)$ and we prove, using the properties of Hecke operators, that $H_S^1(\G,M)$ is a Selmer group containing $\Sel_\LL$. Then, we discuss the time complexity of this method. }

\section*{Introduction}

Selmer groups, constructed from the Galois cohomology of number fields, are powerful tools in modern number theory. Introduced in the study of descent in elliptic curves (\cite[Chapter X, \textsection 4]{Silverman}), they have been crucial for progress toward the BSD conjecture (see for example \cite{Kolyvagin}) and arithmetic statistics on ranks of elliptic curves (see \cite{Bhargava}), conjecturally predict the order of vanishing of L-functions (see \cite{Bloch_Kato}), control deformations of Galois representations (see \cite[\textsection 1.10]{deforming_gal_rep}) and therefore play an important role in modularity theorems (see \cite{Wiles}) and have many other applications, for instance in effective class field theory (see \cite[\textsection 5.2.2]{Cohen}). It is therefore important to design efficient algorithms to compute Selmer groups.

Throughout the article, we will use the following definition of a Selmer group.

 Let $K$ be a number field, $\overline{K}$ its algebraic closure, and $\G$ (or $\G_K$) its absolute Galois group. Let $M$ be a left $\G$-module.
 Given a finite place $v$ of $K$, let $G_v$ denote the decomposition group of $K$ at $v$, and $I_v$ the inertia group. Then, 
 \begin{itemize}
     \item a \emph{local condition} at $v$ is a subgroup $L_v \subset H^1 (G_v, M)$,
     \item the \emph{unramified condition} is the subgroup $$H^1_{un}(G_v, M) = \ker \left\{ H^1(G_v, M) \rightarrow H^1(I_v, M) \right\},$$
     \item a \emph{Selmer system} for $M$ is a set $\LL$ of local conditions $L_v$ at every finite place $v$ of $K$, such that all but finitely many of the $L_v$ are the unramified condition,
     \item given a Selmer system $\LL$, the \emph{Selmer group} attached to $\LL$ is the subgroup of $H^1 ( \G_K, M)$ given by $$\Sel_\LL = \ker \left\{H^1 ( \G_K, M) \rightarrow \prod_{v}\frac{H^1(G_v, M)}{L_v} \right\}.$$
 \end{itemize}

 Note that this definition of Selmer group is restricted to subgroups of the first cohomology group $H^1 ( \G_K, M)$, but we can give a similar definition for Selmer groups that would be subgroups of other cohomology groups. For future research, it might be interesting to try and adapt the method of this article to be able to compute Selmer groups contained in $H^2 ( \G_K, M)$.

Some methods already exist to compute Selmer groups. For Selmer groups of elliptic curves, Bruin lists some of these algorithms in \cite[section 5.4]{Bruin} and gives a geometric interpretation, and we can also mention some more recent articles, like the article \cite{maistret_shukla} by Maistret and Shukla. The method presented here is more general, since it allows one to compute Selmer groups in general and not only Selmer groups of elliptic curves. For future work, we think it would be interesting to compare the time complexity of all the existing methods.

The main result in this article will be the following.
\begin{introtheorem} \label{tmA}
    Let $\G$ be the absolute Galois group of a number field $K$, and $M$ be a finite left $\G$-module. There exists an algorithm that on input
    \begin{itemize}
        \item the module $M$,
        \item the finite group $G$ that is the image of the action $\G \rightarrow \Aut(M)$,
        \item a Selmer system $\LL$,
    \end{itemize}
    outputs the Selmer group $\Sel_\LL$ attached to $\LL$ for $M$.
    Moreover, every step of this algorithm is polynomial, except for the computation of subfields of $\overline{K}$ fixed by subgroups of $\G$, and the computation of the group of $S$-units and the class group of some field extensions of $K$.
\end{introtheorem}

We will describe this algorithm (see algorithm \ref{alg:main}), and discuss the complexity in proposition \ref{pro_complexity}.

We will use some properties of Hecke operators of finite groups. In section \ref{Sec:Hecke}, we will give the definition of Hecke operators, as it was given for example in \cite{Yoshida}, and we will state some important properties that were proven in \cite{clgp} (propositions \ref{pseudo_inv} and \ref{im_S_unit}).

In section \ref{Sec:resolution}, we will discuss a method to find a partial resolution of any finite Galois module $M$, of the form $$ 0 \rightarrow M \rightarrow I_0 \xrightarrow{d_0} I_1 \xrightarrow{d_1} I_2 $$ where the modules $I_i$ are duals of sums of permutation modules, and the morphisms $d_i$ are given by sums of Hecke operators. A more precise algorithm will be given in section \ref{Sec:Algo} (algorithm \ref{alg_resolution}).

Then, in section \ref{Sec:Selm}, we will introduce some special groups $H^1_S(\G,M)$ where $S$ is a set of prime numbers (see definition \ref{df:H1S}), and we will show that under some conditions on $S$, the group $H^1_S$ is the Selmer group attached to the Selmer structure where all the conditions at places outside of $S$ are unramified conditions and where there is no condition for the places in $S$. 

Given a Selmer system $\LL$, our method to compute $\Sel_\LL$ will be to first compute $H^1_S$ for $S$ large enough, and then look for $\Sel_\LL$ as a subgroup of $H^1_S$. In section \ref{Sec:Algo}, we will describe this method (see algorithm \ref{alg:main}) and discuss its time complexity (see proposition \ref{pro_complexity}). We leave the implementation of this method for future work.

\paragraph*{Notations and conventions}

In all of the article, $K$ will be a field of characteristic zero. We will denote by $\overline{K}$ its algebraic closure and by $\G$ the Galois group $\Gal (\overline{K}/K)$.

All modules in this article will be left modules.

When $R$ is a ring and $M$, $N$ are left $R$-modules, we will denote~$\Hom_R(M,N)$ the group of $R$-module homomorphisms from $M$ to $N$.

If $M$ is a $\G$-module, we will denote by $M^* := \Hom_{\Z}(M, \K)$ the dual module of $M$, where $\K$ is viewed as an abelian group.

In a finite field extension $L/F$, we will denote by $N_{L/F}(x)$ the norm of $x \in L$.

Unless specified otherwise, the group laws of cohomology groups will always be denoted multiplicatively. 

\paragraph*{Acknowledgements} I would like to thank A. Page for the the initial idea behind this article and for his precious help and advice.\\
This research was funded by the University of Bordeaux. It was also supported by the CIAO ANR (ANR-19-CE48-008) and the CHARM ANR (ANR-21-CE94-0003). It took place inside the CANARI team (Cryptographic Analysis and Arithmetic) of the Institute of Mathematics of Bordeaux (IMB).\\

\section{Properties of Hecke operators} \label{Sec:Hecke}

The method will use some properties of Hecke operators, which were proven in \cite{clgp}.
In this section, $G$ will denote a finite subgroup of $\G = \Gal(\overline{K}/K)$.

\begin{df}
    Let $H$ be a subgroup of $G$, then the set $\Z[G/H]$ has a structure of $\Z[G]$-module. We call \emph{permutation modules} the sums of modules of this form.
\end{df}

\begin{df} \label{def_coset_Hecke}
    If $R$ is a ring and $V$ a $R[G]$-module, and $H,J$ two subgroups of $G$, then there is a morphism of $R$-modules $$R[H \backslash G / J] \rightarrow \Hom_{R}(V^H, V^J).$$
    This isomorphism is described in \cite[Fact 1.4]{clgp}. The morphisms of $R$-modules associated with cosets of the form $HgJ$ for $g$ in $G$ by this isomorphism are called \emph{Hecke operators}.
\end{df}

\begin{pro} \label{pseudo_inv}
    Let $H_1, \cdots, H_n$ and $J_1, \cdots, J_m$ be subgroups of $G$.
    If $\Phi \colon \bigoplus_i \Z[G/H_i] \rightarrow \bigoplus_i \Z[G/J_i]$ is a morphism of $\Z[G]$-modules whose image is of finite index in $\bigoplus_i \Z[G/J_i]$, then there exists an injective morphism $\Psi \colon \bigoplus_i \Z[G/J_i] \rightarrow \bigoplus_i \Z[G/H_i]$ such that $\Phi \circ \Psi = k \cdot \id$, where $k$ is a positive integer that divides $|G|^2$.
\end{pro}

\begin{proof}
    See \cite[Proposition 2.13]{clgp} for the existence of $\Psi$, and \cite[Theorem 2.16]{clgp} for the proof that $k$ divides $|G|^2$.
\end{proof}

\begin{pro} \label{im_S_unit}
    Let $H,J$ be two subgroups of $G$, $L_1 = \overline{K}^H$ and $L_2 = \overline{K}^J$. Let $S$ be a set of prime numbers. If $\Phi: L_1^\times \rightarrow L_2^\times$ is defined by a sum of Hecke operators, then $$\Phi( \Z_{S, L_1}^\times) \subset \Z_{S, L_2}^\times. $$
\end{pro}

\begin{proof}

This is a direct consequence of \cite[Theorem 1.18]{clgp}.

\end{proof}

\section{Finding a resolution with Hecke operators} \label{Sec:resolution}

In all of the article, $M$ will be a finite Galois module.

Let $G$ be the image of the action $\G \rightarrow \Aut(M)$. It is isomorphic to a finite quotient of $\G$.
Note that the action of $\G$ over $M$ can be factorized to be seen as an action of $G$ over $M$.

\begin{req} \label{req:G_galois}
    If $\mathcal{N}$ denotes the kernel of the action $\G \rightarrow \Aut(M)$, then $G$ is the Galois group of the Galois extension $\overline{K}^{\mathcal{N}}/K$.
\end{req}

Suppose we have $\Z[G]$-modules $P_i$ for every integer $i$, which are permutation modules, as well as some morphisms of $G$-modules $s$ and  $d_i^*$ such that the sequence $$\cdots \xrightarrow{d_2^*} P_2 \xrightarrow{d_1^*} P_1 \xrightarrow{d_0^*} P_0 \xrightarrow{s} M^* \rightarrow 0 $$ is exact, where $M^*$ is the dual module of $M$.

We will see in section \ref{Sec:Algo} that we can always find such an exact sequence, and we will give an algorithm (algorithm \ref{alg_resolution}) to compute such $P_i$ and $d_i^*$ up to any integer $i$.

In this article, we will only need to compute such sequences up to~$P_2$. We will denote by (\ref{suite_ex_1}) an exact sequence of the form \begin{equation} \label{suite_ex_1}
     P_2 \xrightarrow{d_1^*} P_1 \xrightarrow{d_0^*} P_0 \xrightarrow{s} M^* \rightarrow 0. \end{equation}

\begin{lm} \label{lm_ex_fct}
    The functor $P \mapsto P^* = \Hom_\Z(P, \K)$, on the category of $\G$-modules that are finitely generated $\Z$-modules, is exact.
\end{lm}

\begin{proof}

Let $A, B, C$ be $G$-modules such that there is an exact sequence $$ 0 \rightarrow A \xrightarrow{i} B \xrightarrow{s} C \rightarrow 0.$$

Let $s^*$ be the map $C^* = \Hom(C, \K) \rightarrow \Hom(B, \K) = B^*; f \mapsto f \circ s$. Since $s$ is surjective, the map $s^*$ is injective.

Likewise, let $i^*$ be the map $B^* \rightarrow A^*, f \mapsto f \circ i$. Again, since $i$ is injective, the map $i^*$ is surjective.

Since we know that $i \circ s = 0$, we have $\im(s^*) \subseteq \Ker(i^*)$ . We also need to show that $\Ker(i^*) \subseteq \im(s^*)$.

Let $f$ be an element of $\Ker(i^*)$. That is to say $f \circ i = 1$. Then that means $\Ker(f) \subset \im(i) = \Ker(s)$. So, by the structure of finitely generated abelian groups, there exists $g \in C^*$ such that $f = g \circ s$.

In conclusion, the short sequence $$0 \rightarrow C^* \xrightarrow{s^*} B^* \xrightarrow{i^*} A^*\rightarrow 0$$ is exact.

\end{proof}

Once we obtain an exact sequence of the form (\ref{suite_ex_1}), by lemma ~\ref{lm_ex_fct}, we can take the dual and get an exact sequence of the form \begin{equation} \label{suite_ex_2}
     0 \rightarrow M \rightarrow I_0 \xrightarrow{d_0} I_1 \xrightarrow{d_1} I_2 \end{equation}

where $I_i = P_i^*$ for all $i$.

Consider an exact sequence of the form (\ref{suite_ex_2}) obtained with the construction described above. The modules $P_0$, $P_1$, $P_2$ are permutations modules. In the rest of the section, let us denote $P_i = \bigoplus_j \Z[G/H_{i,j}]$ for $i \in \{1, 2, 3\}$. And for every pair $(i,j)$, let us define $L_{i,j} := \overline{K}^{H_{i,j}}$.

\begin{pro} \label{pro:alg_étales}
    With the above notations, for $i \in  \{1, 2, 3\}$, we have $$I_i = \bigoplus_j \Ind_{\G / G_{L_{i,j}}} \overline{K}^\times = \bigoplus_j \overline{L_{i,j}}^\times$$
    where $G_{L_{i,j}}$ is the absolute Galois group of $L_{i,j}$.
    And $$I_i^G = \bigoplus_j L_{i,j}^\times.$$
\end{pro}

\begin{proof}
    See \cite[Section 3.12, Example 19]{Voskresenskii}.
\end{proof}

The morphisms $d_0 \colon I_0 \rightarrow I_1$ and $d_1 \colon I_1 \rightarrow I_2$ induce some morphisms respectively from $I_0^G$ to $I_1^G$ and from $I_1^G$ to $I_2^G$, that we will denote $d_0^G$ and $d_1^G$.

\begin{pro} \label{coh_grp}
    With the above notations, we have $$H^1(\G,M) = \frac{\Ker (d_1^G \colon I_1^G \rightarrow I_2^G ) }{\im (d_0^G \colon I_0^G \rightarrow I_1^G)}.$$
\end{pro}

\begin{proof}

Let $J \subset I_1$ be the image of $d_0$. Then we have a short exact sequence $$ 0 \rightarrow M \rightarrow I_0 \xrightarrow{d_0} J \rightarrow 0. $$

The associated long exact sequence starts with $$ 0 \rightarrow M^\G \rightarrow I_0^\G \xrightarrow{d_0} J \rightarrow H^1(\G, M) \rightarrow H^1(\G, I_0)  $$
and $H^1(\G, I_0) =  \bigoplus_j H^1(G_{L_{0,j}}, \overline{L}_{0,j}^\times)$ by Shapiro's lemma, and  $  H^1(G_{L_{0,j}},\overline{L}_{0,j}^\times) = 0 $ by Hilbert 90th theorem.

This last exact sequence allows us to deduce that $$H^1(\G, M) = \frac{J^G}{\im (d_0^G \colon I_0^G \rightarrow I_1^G)}. ~~ (*)$$

What's more, by definition of $J$, we also have an exact sequence $$0 \rightarrow J \rightarrow I_1 \xrightarrow{d_1} I_2.$$

Hence the exact sequence $$0 \rightarrow J^G \rightarrow I_1^G \xrightarrow{d_1} I_2^G$$
from which we can deduce that $$J^G = \Ker (d_1^G \colon I_1^G \rightarrow I_2^G ).$$ 

Combining that result with $(*)$, we get $$H^1(\G,M) = \frac{\Ker (d_1^G \colon I_1^G \rightarrow I_2^G ) }{\im (d_0^G \colon I_0^G \rightarrow I_1^G)}.$$

\end{proof}

\begin{pro} \label{pro:res_compatible}
    For every subgroup $\HH < \G_K$, the map $$\Res: H^1(\G_K, M) \rightarrow H^1(\HH, M)$$ is the natural restriction $$H^1(\G, M) = \frac{\Ker (d_1^\G \colon I_1^\G \rightarrow I_2^\G ) }{\im (d_0^\G \colon I_0^\G \rightarrow I_1^\G)} \rightarrow H^1(\HH, M) = \frac{\Ker (d_1^\HH \colon I_1^\HH \rightarrow I_2^\HH ) }{\im (d_0^\HH \colon I_0^\HH \rightarrow I_1^\HH)}$$.
\end{pro}

\begin{proof}

Let $J \subset I_1$ be the image  of $d_0$. Then we have a short exact sequence $$ 0 \rightarrow M \rightarrow I_0 \xrightarrow{d_0} J \rightarrow 0. $$

The associated long exact sequence starts with $$0 \rightarrow M^\G \rightarrow I_0^\G \rightarrow J_1^\G \rightarrow H^1(\G, M).$$

We can then apply the restriction map to obtain 

\[ 
\xymatrix@C4.5pc{
0 \ar[r] & M^\G \ar[r] \ar[d]^{\Res} &  I_0^\G \ar[r] \ar[d]^{\Res} & J_1^\G \ar[r] \ar[d]^{\Res} & H^1(\G, M) \ar[d]^{\Res} \\
0 \ar[r] & M^{\HH} \ar[r] & I_0^{\HH} \ar[r] & J_1^{\HH} \ar[r] & H^1({\HH}, M)
}
\] 

Moreover, for every field $F$ such that $K \subset F$, and for all $i$, we have $I_i ^\G = L_i^\times$ and $I_i^{\Gal(\overline{F}/F)} = (L_i \otimes_K F)^\times$. So $I_i^{\HH} = \overline{L_i}^\HH$, hence the conclusion.

\end{proof}

\section{A remarkable Selmer group}
\label{Sec:Selm}

Let $M$ be a finite Galois module, suppose we have the Galois modules $I_0$, $I_1$, $I_2$ and the morphisms of $G$-modules $d_0$ and $d_1$ obtained as in section \ref{Sec:resolution}, such that the sequence $$ 0 \rightarrow M \rightarrow I_0 \xrightarrow{d_0} I_1 \xrightarrow{d_1} I_2 $$ is exact.
By proposition \ref{coh_grp}, we have $$H^1(\G,M) = \frac{\Ker (d_1^G \colon I_1^G \rightarrow I_2^G ) }{\im (d_0^G \colon I_0^G \rightarrow I_1^G)}.$$ where for all $i \in \{0, 1, 2\}$, $I_i^G$ is of the form $\bigoplus_j L_{i,j}^\times$ and the $L_{i,j}$ are intermediary fields between $K$ and $\overline{K}$.

In the rest of the article, we will denote by $L_i$ the étale algebra $\prod_j L_{i,j}$. We will allow ourself to extend to étale algebras the notions of class groups and $S$-unit groups.

\begin{df} \label{df:H1S}

Let $S$ be a set of prime numbers. Let us define the group $H^1_S(\G, M) := \frac{\Ker( d_1^G \colon \bigoplus_j \Z_{L_{1,j},S}^\times \rightarrow \bigoplus_j \Z_{L_{2,j},S}^\times) }{\im ( d_0^G \colon \bigoplus_j \Z_{L_{0,j},S}^\times \rightarrow \bigoplus_j \Z_{L_{1,j},S}^\times)}$.

\end{df}

By proposition \ref{im_S_unit}, the images of $S$-units by $d_1^G$ and $d_0^G$ are $S$-units, so the group $H^1_S(\G, M)$ is well defined.

When the context is clear, we will also allow ourself to write $H^1$ and $H^1_S$ instead of $H^1(\G,M)$ and $H^1_S(\G,M)$.

The goal of this section will be to prove that $H^1_S(\G, M)$ is a Selmer group.
We will use the following notations:
\begin{itemize}
    \item $Z^1 := \Ker (d_1^G \colon I_1^G \rightarrow I_2^G) $
    \item $B^1 := \im (d_0^G \colon I_0^G \rightarrow I_1^G)$
    \item $Z^1_S := \Ker( d_1^G \colon \bigoplus_j \Z_{L_{1,j},S}^\times \rightarrow \bigoplus_j \Z_{L_{2,j},S}^\times)$
    \item $B^1_S := \im ( d_0^G \colon \bigoplus_j \Z_{L_{0,j},S}^\times \rightarrow \bigoplus_j \Z_{L_{1,j},S}^\times)$ .
\end{itemize}

\begin{pro}
We have an injection $H^1_S(\G, M) \hookrightarrow H^1(\G, M)$ 
\end{pro}

\begin{proof}

We have trivially $Z^1_S \subset Z^1$. So in order to prove the proposition, it is enough to show that $B^1 \cap \Z_{L_1, S}^\times \subset B^1_S $.
In other words, we need to show that if an element $y$ in the image of $d_0^G$ is an $S$-unit, then there exists an $S$-unit $x$ in $I_0^G = L_0^\times$ such that $d_0^G(x) = y$.\

If we take the tensor product of the exact sequence (\ref{suite_ex_2}) with $\Q$, we obtain 
$$0 \rightarrow I_0 \otimes \Q \xrightarrow{d_0} I_1 \otimes \Q \xrightarrow{d_1} I_2 \otimes \Q$$
because $M$ is finite, so that $M \otimes \Q = 0$.

Then, by proposition \ref{pseudo_inv}, there exists a surjective morphism of $G$-modules $f \colon I_1 \rightarrow I_0$ such that $f \circ d_0 = k \cdot \id$, with $k$ dividing $|G|^2$.

Now, let $y$ be an element of $B^1 \cap \Z_{L_1, S}^\times$. Since $x$ is in $B^1$, there exists $x \in L_0^\times$ such that $d_0^G(x) = y$. So $s \circ d_0^G (x) = k \cdot x$ (or $x^k$ in multiplicative notation). But $d_0^G(x) = y$ is an $S$-unit, so its image by $f$ is also an $S$-unit by \ref{im_S_unit}.
Hence $k \cdot x$ is an $S$-unit. And since $\Z_{S, L_0}^\times$ is saturated as a subgroup of $L_0^\times$, that means $x$ is also an $S$-unit.

So $y$ is the image of an $S$-unit by $d_0^G$, ie $y \in B_S^1$. Hence $B^1_S \subset B^1 \cap \Z_{L_1, S}^\times$.

\end{proof}

\begin{df} \label{df:notations}

In the rest of the section, if $v$ is a finite place of $K$, we will use the following notations:
\begin{itemize}
    \item $Z^1_{\units, v} = \Ker(\Z_{L_1,v}^\times \xrightarrow{d_1} \Z_{L_2,v}^\times)$,
    \item $B^1_{\units, v} = \im(\Z_{L_0,v}^\times \xrightarrow{d_0} \Z_{L_1,v}^\times)$,
    \item$H^1_{\units, v} = \frac{Z^1_{\units, v}}{B^1_{\units, v}}$,
    \item $K_v^{\ur}$ the largest unramified extension of $K_v$, and $I_{K_v} = \Gal(\overline{K_v}/K_v^{\ur})$ the inertia group, and $\G_{K_v} = \Gal(\overline{K_v}/K)$.
    \item $Z^1_{\ram, v} = \Ker((L_0 \otimes K_v^{\ur})^\times \xrightarrow{d_1} (L_1 \otimes K_v^{\ur})^\times)$,
    \item $B^1_{\ram, v} = \im((L_1 \otimes K_v^{\ur})^\times \xrightarrow{d_0} (L_2 \otimes K_v^{\ur})^\times)$,
    \item $H^1_{\ram,v} = \frac{Z^1_{\ram, v}}{B^1_{\ram, v}}$.

\end{itemize}

When the context is clear, we will allow ourself to write $H^1_{\ram}$, $Z^1_{\ram}$ and $B^1_{\ram}$.

For every place $v$ of $K$, we also denote:
\begin{itemize}
    \item $Z^1_v = \Ker(L_{1,v}^\times \xrightarrow{d_{1,v}} L_{2,v}^\times)$
    \item $B^1_v = \im(L_{0,v}^\times \xrightarrow{d_{0,v}} L_{1,v}^\times)$
    \item $H^1_v = \frac{Z^1_v}{B^1_v}$.
\end{itemize}

where $d_{0,v}$ and $d_{1,v}$ are defined respectively by the two 
following commutative diagrams

\begin{tikzpicture}[->,>=stealth',shorten >=2pt,auto,semithick]
    \node (A) {$L_0^\times$};
    \node (B) [right of=A, xshift=2cm] {$L_1^\times$};
    \node (C) [below of=A,, yshift = -1cm] {$L_{0,v}^\times$};
    \node (D) [below of=B, yshift = -1cm] {$L_{1,v}^\times$};
    \path (A) edge [] node [above] {$d_0$} (B);
    \path (A) edge [] (C);
    \path (B) edge [] (D);
    \path (C) edge [] node [below] {$d_{0,v}$} (D);

    \node (E) [right of=B, xshift= 1.5cm] {$L_1^\times$};
    \node (F) [right of=E, xshift=2cm] {$L_2^\times$};
    \node (G) [below of=E, yshift = -1cm] {$L_{1,v}^\times$};
    \node (H) [below of=F, yshift = -1cm] {$L_{2,v}^\times$};
    \path (E) edge [] node [above] {$d_1$} (F);
    \path (E) edge [] (G);
    \path (F) edge [] (H);
    \path (G) edge [] node [below] {$d_{1,v}$} (H);

    \node(I) [left of = A, xshift = 0.3cm, yshift = -1cm] {$(\star)$};
    \node(J) [left of = E, xshift = 0.3cm, yshift = -1cm] {$(\star \star)$};
\end{tikzpicture}

\end{df}

\begin{pro}
    With the notations of definition \ref{df:notations}, we have $H^1_v = H^1(\G_{K_v}, M)$ and $H^1_{\ram} = H^1(I_{K_v}, M)$.
\end{pro}

\begin{proof}
    We can do the same construction as in section \ref{Sec:resolution}, with $K_v$ (respectively $K^{\ur}_v$) instead of $K$. The same resolution $$ 0 \rightarrow M \rightarrow I_0 \xrightarrow{d_0} I_1 \xrightarrow{d_1} I_2 $$ also works in these cases, since the $I_i$ are also both $\G_{K_v}$-modules and $I_{K_v}$ modules. Moreover, for all $i$, we have $I_i^{\G_{K_v}} = L_{i,v}^\times$ and $I_i^{I_{K_v}} = (L_i \otimes K_v^{\ur})^\times$. Hence the conclusion.
\end{proof}

Note that, by proposition \ref{pro:res_compatible}, the map $\Res_v: H^1(\G_K, M) \rightarrow H_{1,v}$ is the natural restriction $\frac{Z_1}{B_1} \rightarrow \frac{Z_{1,v}}{B_{1,v}}$.

\begin{lm} \label{incl}
Let $v \notin S$ be a place of $K$, then
we have $$\Res_v( H^1_S(\G,M) ) \subseteq H^1_{\units,v}.$$
\end{lm}

\begin{proof}

Let $v$ be a place not in $S$.
Let $\overline{x}$ be a class in $H_S^1$, $x \in Z_S^1$ a representative of $\overline{x}$, and $x_v$ the localisation of $x$ at $v$. Since $v \notin S$, we have $x_v \in \Z_{L_1, v}^\times$. And since the diagram $(\star \star)$ commutes, we have $x_v \in Z^1_v$. So $x_v \in \Ker(\Z_{L_1, v}^\times \xrightarrow{d_{1,v}} L_{2,v}^\times)$. And if $x_2 = x \cdot b$ is another representative of $\overline{x}$, with $b \in B^1_S$, then it is easy to check that $b_v \in \im(\Z_{L_0,v}^\times \rightarrow \Z_{L_1,v}^\times)$. Hence $\Res_v(\overline{x}) \in H_{\units, v}$.

\end{proof}

\begin{pro} \label{pro:selmer1}
If $S$ is a set of primes such that the class group $\Cl(L_0)$ is spanned by ideals above all primes in $S$, then
we have $$H^1_S = \{ x \in H^1 \mid \forall v \notin S, \Res_v (x) \in H^1_{\units,v} \}.$$

\end{pro}

\begin{proof}

By lemma \ref{incl}, we already have the inclusion $H^1_S \subseteq \{ x \in H^1 \mid \forall v \notin S, \Res_v (x) \in H^1_{\units,v} \}$. 

Now, let $\overline{x}$ be a class in $H^1(\G, M)$ such that for all place $v \notin S$, we have $\Res_v(x) \in H^1_{\units, v}$.
By definition, for all $v$, there exists $z_v$ in $L_{0,v}^\times$ such that $\Res_v(x) \cdot d_{0,v}(z_v) \in \Z_{L_1, v}^\times$.

We want to show that there exists $z \in L_0^\times$ such that for all $v \notin S$, $z \cdot z_v^{-1} \in \Z_{L_0, v}^\times$. 
This problem is equivalent to taking a fractional ideal $\mathfrak{a}$ of $L_0$, and deciding whether there exists $\mathfrak{p}$ a product of prime ideals in $S$ such that $\mathfrak{a} \mathfrak{p}$ is principal. But since $S$ spans the class group of~$L_0$, we know it is possible.

\end{proof}

\begin{pro} \label{pro:selmer2}
For every place $v$ of $K$ such that $v$ does not divide $|M|$, we have $H^1_{\units, v} = \Ker (\Res \colon H^1 \rightarrow H^1_{\ram})$.
\end{pro}

\begin{proof}
First let us show the inclusion $H^1_{\units,v} \subseteq \Ker (\Res \colon H^1 \rightarrow H^1_{\ram})$.

We have the following diagram:

\[ 
\xymatrix@C4.5pc{
(\O_{L_0} \otimes \O_{K_v^{\ur}}) ^\times \ar[r]^{d_0} \ar@{^{(}->}[d]^{i} & (\O_{L_1} \otimes \O_{K_v^{\ur}}) ^\times \ar[r]^{d_1} \ar@{^{(}->}[d]^{i} &  (\O_{L_2} \otimes \O_{K_v^{\ur}}) ^\times \ar@{^{(}->}[d]^{i}\\ 
(L_0 \otimes K_v^{\ur}) ^\times \ar[r]^{d_0} \ar@{->>}[d]^{\val} & (L_1 \otimes K_v^{\ur}) ^\times \ar[r]^{d_1} \ar@{->>}[d]^{\val} &  (L_2 \otimes K_v^{\ur}) ^\times \ar@{->>}[d]^{\val}\\
\Z^{\Hom(L_0, K_v^{\ur})} \ar@{^{(}->}[r]^{d_0} & \Z^{\Hom(L_1, K_v^{\ur})} \ar[r]^{d_1} & \Z^{\Hom(L_2, K_v^{\ur})}
}
\] \label{diag}

where the three vertical sequences are exact, and where $\val$ denotes the valuation morphisms. What's more, the morphism $d_0 \colon \Z^{\Hom(L_0, K_v^{\ur})} \rightarrow \Z^{\Hom(L_1, K_v^{\ur})}$ is injective because the kernel of $d_0 \colon (L_0 \otimes K_v^{\ur}) ^\times \rightarrow (L_1 \otimes K_v^{\ur}) ^\times$ is torsion, so its image under $\val$ is $0$.\\

Let $x \in Z^1_{\units} \subset \O_{L_1}^\times$. That is to say $d_1(x) = 1 \in \O_{L_2}^\times$.
We want to show that $\Res(x) = x \otimes 1 \in (L_1 \otimes K_v^{\ur} )^\times$ is in $B^1_{\ram} = d_0( (L_0 \otimes K_v^{\ur})^\times)$.

Let $N$ be an integer such that the module $M$ is annihilated by $N$, and such that $N$ does not divide the characteristic of the residue field of $\O_K$. Then $H^1_{\ram}$ is $N$-torsion.

So there exists $y \in (L_0 \otimes K_v^{\ur})^\times$ such that $\Res(x) ^N = d_0(y)$. 

Since $x \in \O_{L_1}^\times$, then $\val( \Res(x) ) = 0$, so $\val (y) = 0$, so $y \in (\O_{L_0}^\times \otimes \O_{K_v^{\ur}}^\times) ^\times$. And $(\O_{L_0}^\times \otimes \O_{K_v^{\ur}}^\times) ^\times$ is $N$-divisible, so there exists $z \in (\O_{L_0}^\times \otimes \O_{K_v^{\ur}}^\times) ^\times$ such that $z^N = y$.

Hence $d_0(z)^N = d_0(y) = x^N$. This proves that $d_0(z) = \zeta_N x$, with $\zeta_N$ a $N$-th root of unity.\\
Now let us prove that for the $N$-th roots of unity, $\im(d_0) = \Ker(d_1)$, which would imply the conclusion.\\

With the notation of section \ref{Sec:resolution}, we have an exact sequence $$ P_2 \xrightarrow{d_1^*} P_1 \xrightarrow{d_0^*} P_0 \xrightarrow{s} M^* \rightarrow 0. $$

Consider the modules $P_2' = \im(d_1^*)$ and $P_0' = \im(d_0^*)$. Then, by definition, we have the short exact sequence $$0 \rightarrow P_2' \rightarrow P_1 \rightarrow P_0' \rightarrow 0.$$

By properties of Tor functors (see for example \cite[chapter VI] {Hom_algebra}), and because the modules $P_1, P_0', P_2'$ are $N$-torsion free, we have the short exact sequence $$0 \rightarrow P_2'/N \rightarrow P_1/N \rightarrow P_0'/N \rightarrow 0.$$
By taking the dual, we then get precisely that $\im(d_0) = \Ker(d_1)$ for the $N$-th roots of unity, because for every $i$, we have $(P_i/N)^* = I_i [N]$, and $I_i[N]$ is the set of $N$-th roots of unity of $\overline{L_i}$.\\

Now let us show the second inclusion: $H^1_{\units, v} \supseteq \Ker (\Res \colon H^1 \rightarrow H^1_{\ram})$.

Let $x$ be an element of $\Ker (\Res \colon H^1 \rightarrow H^1_{\ram})$, that is to say an element of $L_1^\times$ such that $\Res(x) \in B^1_{\ram}$. We want to show that there exists $z \in B^1$ such that $x \cdot z^{-1} \in \O_{L_1}^\times$.

As $\Res(x)$ is in $B^1_{\ram}$, there exists $y \in (L_0 \otimes K_v^{\ur})^\times$ such that $d_0( \val(y) ) = \val(\Res(x))$. Besides, since $x$ is in $L_1^\times$, then $\val (\Res(x))$ is invariant by the action of $\Gal(K_v^{\ur}/K)$.

So for all $g \in \Gal(K_v^{\ur}/K)$, $g \cdot d_0(\val(y)) = d_0(\val(y) ) = d_0(g \cdot \val(y) )$. 
Since  $d_0 \colon \Z^{\Hom(L_0, K_v^{\ur})} \rightarrow \Z^{\Hom(L_1, K_v^{\ur})}$ is injective, that means $\val(y)$ is also invariant by the action of $\Gal(K_v^{\ur}/K)$.

Therefore, $\val(y)$ is in $ (\Z^{\Hom(L_0, K_v^{\ur})})^{(\Gal(K_v^{\ur}/K))}$, so there exists $z \in L_0^\times$ such that $\val(z) = \val(y)$.

And thus $\val( \Res(d_0(z) ) ) = d_0(\val(z)) = d_0(\val(y) ) = \val( \Res(x) )$, hence $\val( \Res( d_0(z) x^{-1}) ) = 0$.

So, again by injectivity, $d_0(z) = x$, hence the conclusion.

\end{proof}

\begin{tm} \label{tm:H1S_selmer}
    If all prime ideals above $S$ span $\Cl(L_0)$  and $S$ contains all primes that divide $|M|$, then $H^1_S$ is a Selmer group.  More precisely, it is the Selmer group attached to the Selmer structure where all the conditions at places outside of $S$ are unramified conditions and where there is no condition for the places in $S$.
\end{tm}

\begin{proof}

    The theorem is a direct consequence of proposition \ref{pro:selmer1} and proposition \ref{pro:selmer2}.

\end{proof}

\begin{req}
    Since every Selmer group is contained in a $H^1_S$ for some finite set of places $S$, this gives another proof that Selmer groups are finitely generated.
\end{req}

\section{Algorithm and complexity} \label{Sec:Algo}

In this section, we will explain the algorithmic method to obtain a partial resolution of a finite Galois module $M$, with permutation modules, as discussed in section \ref{Sec:resolution} (See algorithm \ref{alg_resolution}). Then, we will describe the algorithm to compute Selmer groups, (see algorithm \ref{alg:main}) and discuss its complexity (see proposition \ref{pro_complexity}). 

But first, we have to explain how to represent in bits all the mathematical objects involved.

Let $M$ be a finite Galois module, and $G$ be the image of the action $\G \rightarrow \Aut(M)$. It is a finite group, so we can represent it as a subgroup of a permutation group. We can also suppose that we have a list $[g_1, \cdots, g_r]$ of generators.

Since $M$ is a finite module, we can represent it as a list $[m_1, \cdots, m_s]$ of generators of $M$ as an abelian group, and a list of relations, as well as a list of matrices giving the actions of the generators of $G$ on the~$m_i$.  

We can represent a Selmer system $\LL$, with a set of primes, indicating the places where the local conditions are not the unramified condition, a basis of the local cohomology groups at these places and the generators of the subgroups in $\LL$.

As for the Selmer group $\Sel_{\LL}$, since it is a finitely generated group, we can represent it as a list of generators and a list of relations, or by its decomposition in cyclic factors, with the theorem of structure of finitely generated abelian groups.

\begin{algo} \label{alg_perm_surj} ~~\\

\underline{input:} A finite group $G$ and a finitely generated $G$-module $N$.\\
\noindent \underline{output:} A permutation module $P$ and a surjective morphism of $G$-modules $s \colon P \rightarrow N$.

\begin{itemize}
    \item Let $(x_1, \cdots x_r)$ be a generating sequence of elements of $N$.
    \item \label{Hx_et_f_x} For every element $x$ in $\{x_1, \cdots, x_r\}$, 
        \begin{itemize}
            \item compute $H_x = \Stab_{G}(x)$ the stabilizer of $x$ under the action of $G$.
            \item Compute $f_x \colon \Z[G/H_x] \rightarrow N, ~ 1\cdot H_x \mapsto x$.
        \end{itemize}
    \item \label{P0_et_s} Return $P = \bigoplus_{i = 1}^r \Z[G/H_{x_i}]$ and $s = \sum_{i = 1}^r f_{x_i}$.
\end{itemize}

\end{algo}

\begin{algo} \label{alg_resolution}
~~\\
\underline{input:} A finite Galois module $M$, of Galois group $\G$, and $G$ the image of the action $\G \rightarrow \Aut(M)$.

\noindent \underline{output:} Permutation modules $P_i$ and morphisms of $G$-modules $s$ and  $d_i^*$  such that the sequence $$\cdots \xrightarrow{d_2^*} P_2 \xrightarrow{d_1^*} P_1 \xrightarrow{d_0^*} P_0 \xrightarrow{s} M^* \rightarrow 0 $$ is exact.

\begin{enumerate}
    \item Compute $M^*$, take $(x_1, \cdots, x_r)$ a finite generating sequence of elements of $M^*$.
    \item Using algorithm \ref{alg_perm_surj}, compute a permutation module $P_0$ as well as a surjective morphism of $\G$-module  $s \colon P_0 \rightarrow M^*$
    \item Compute the kernel $K_0$ of $s$.
    \item Use algorithm \ref{alg_perm_surj} again, on $K_0$, to obtain $P_1$ and $d_0^*$.
    \item Repeat the same process again to obtain all the $P_i$ and the $d_i^*$.
\end{enumerate}

\end{algo}

Suppose we have a Selmer system $\LL$, and we want to compute $\Sel_{\LL}$, the Selmer group attached to $\LL$ for $M$. Using the results in part \ref{Sec:resolution} and \ref{Sec:Selm}, we deduce the following algorithm.

\begin{algo} \label{alg:main}
~~\\
\underline{input:} A finite Galois module $M$, of Galois group $\G$, and $G$ the image of the action $\G \rightarrow \Aut(M)$. A Selmer system $\LL$.

\noindent \underline{output:} The Selmer group group $\Sel_{\LL}$ 
\end{algo}

\begin{itemize}
    \item Use algorithm \ref{alg_resolution} to compute a resolution of $M$ as in section \ref{Sec:resolution}.
    \item Let $S$ be the smallest set of primes such that all conditions in $\LL$ outside of $S$ are the unramified condition and such that $S$ spans the class group $\Cl(L_0)$ and $S$ contains all the primes that divide $|M|$.
    \item Compute $H_S^1(G,M)$. 
    \item Look for $\Sel_{\LL}$ as a subgroup of the finitely generated group $H_S^1(G,M)$.
\end{itemize}

\begin{tm}
    The algorithms \ref{alg_perm_surj}, \ref{alg_resolution} and \ref{alg:main} are correct.
\end{tm}

\begin{proof}

The correctness of algorithms \ref{alg_perm_surj} and \ref{alg_resolution} are self explanatory, and the correctness of algorithm \ref{alg:main} is a consequence of theorem \ref{tm:H1S_selmer}.

\end{proof}

\begin{pro} \label{pro_complexity}
    If we suppose that we have an oracle that can give us the $S$-units and the class group of any number fields, and another that can compute the fixed field of a subgroup of a Galois group, then the algorithm \ref{alg:main} as a time complexity polynomial in the size of the input and in $|M|$.
\end{pro}

\begin{proof}

First, let us prove that algorithm \ref{alg_resolution} has a time complexity polynomial in the size of $M$ and $G$.
    
\begin{itemize}
    \item If we have a finite $G$-module $M$ given by a list of generators $[m_1, \cdots, m_s]$ and a list of matrices $[M_1, \cdots, M_r]$, as described above, then we can represent the dual module $M^*$ by taking the inverse transpose of all the matrices, twisted by the cyclotomic character $\chi_{|M|}$. 
    
    Indeed, all elements of $M$ are $|M|$-torsion, so $\Hom_\Z(M, \K) = \Hom_\Z(M, \mu_{|M|})$ where $\mu_{|M|}$ is $\Z / |M|\Z$ as a $\G$-module where the action of $\G$ is given by the cyclotomic character $\chi_{|M|}$. So
    \[
    \Hom_\Z(M, \K) = \Hom_\Z(M, \Z / |M|\Z) \otimes \chi_{|M|},
    \]
    and the dual module $M^*$ is computed in polynomial time.

    \item We can compute the stabilizers $\Stab_G(x)$ in time polynomial in the size of $M$, using the method described in \cite[Chapter 4.1]{comp_gp_th}.

    \item With the notations of section \ref{Sec:resolution}, the $P_i$ are all free, finitely generated, $\Z$-modules, they can be represented as in \cite[section 7.4.1]{comp_gp_th}.

    If $[g_1, \cdots, g_r]$ is a list of generators of $G$, let $i$ be an integer, and let us fix $(p_{i,1}, \cdots, p_{i,d_i})$ be a $\Z$-basis of $P_i$. Then we can represent $P_i$ as a list $[\alpha_1, \cdots, \alpha_r]$ where the $\alpha_j$ are the $(d_i \times d_i)$-matrices of the actions of the $g_i$ on the basis $(p_{i,1}, \cdots, p_{i,d_i})$. So their size is still polynomial in the size of the input.

    And the morphisms of $G$-modules $d_0$ and $d_1$ can be represented as a list of co-sets, corresponding to their decompositions in Hecke operators (see definition \ref{def_coset_Hecke}).
    
\end{itemize}

Once we apply algorithm \ref{alg_resolution}, we obtain an exact sequence of the form $$ P_2 \xrightarrow{d_1} P_1 \xrightarrow{d_0} P_0 \xrightarrow{s} M^* \rightarrow 0 $$
and we represent $d_0$ and $d_1$ as a sum of cosets corresponding to Hecke operators. Then, with the notations of proposition \ref{pro:alg_étales}, we can compute the number fields $L_{i,j} = \overline{K}^{H_{i,j}}$ thanks to the oracle.

Then, assuming the oracle gives us the $S$-units of all the $L_{i,j}$, with $S$ easily accessible from the representation of the Selmer system $\LL$ and from our oracle, computing the group $H_S^1(\G,M)$ boils down to computing the actions of Hecke operators on $S$-units, which takes polynomial time (see \cite[Theorem 1.18]{clgp}).

Finally, all there is left to do is to find a basis of $\Sel_\LL$ as a subgroup of $H^1_S(G,M)$. This comes down to computing the kernel of the map $$H_S^1 ( G, M) \rightarrow \prod_{v}\frac{H^1(G_v, M)}{L_v}.$$ 


\end{proof}

\begin{req}
    To compute the fixed fields $L_{i,j} = \overline{K}^{H_{i,j}}$, one can use \cite[algorithm 1]{fixed_fields}. However, the author of the present paper was unable to find a result in the literature about the complexity of this algorithm.
\end{req}

\printbibliography[title = References]

\end{document}